\documentclass[11pt]{article}

\usepackage[T1]{fontenc}              
\usepackage[latin9]{inputenc}         

\usepackage[a4paper]{geometry}        

\usepackage{url}

\usepackage{booktabs}

\usepackage{graphicx} 
\usepackage{amssymb,amsfonts,amsmath}
\usepackage{amsthm}

\usepackage{natbib}
\usepackage{mathpazo}
\usepackage{microtype}

\usepackage{color}

\usepackage{todonotes}
 
\newcommand\bA{\boldsymbol A}
\newcommand\bB{\boldsymbol B}
\newcommand\bC{\boldsymbol C}
\newcommand\bD{\boldsymbol D}
\newcommand\bG{\boldsymbol G}
\newcommand\bh{\boldsymbol h}
\newcommand\bI{\boldsymbol I}

\newcommand\bL{\boldsymbol L}

\newcommand\bQ{\boldsymbol Q}
\newcommand\bU{\boldsymbol U}
\newcommand\bV{\boldsymbol V}
\newcommand\bW{\boldsymbol W}
\newcommand\bX{\boldsymbol X}
\newcommand\bx{\boldsymbol x}

\newcommand\bZ{\boldsymbol Z}
\newcommand\bz{\boldsymbol z}
\newcommand\bS{\boldsymbol S}

\newcommand\bLambda{\boldsymbol \varLambda}
\newcommand\bmu{\boldsymbol \mu}
\newcommand\bPhi{\boldsymbol \varPhi}
\newcommand\bPsi{\boldsymbol \varPsi}
\newcommand\bRho{\boldsymbol P}
\newcommand\bSigma{\boldsymbol \varSigma}
\newcommand\bTheta{\boldsymbol \varTheta}

\newcommand\cor{\text{cor}}
\newcommand\cov{\text{cov}}
\newcommand\diag{\text{diag}}
\newcommand\expect{\text{E}}

\newcommand\trace{\text{tr}}
\newcommand\var{\text{var}}

\newtheorem{prop}{Proposition}

\newcommand{\eqcite}[1]{Equation~(\ref{#1})}
\newcommand{\tabcite}[1]{Table~(\ref{#1})}

\definecolor{revision}{rgb}{1,0,0} 

\begin{document}

\date{2 December 2015; last revision 17 December 2016}

\title{
Optimal Whitening and Decorrelation
}

\newcommand\myaddress[2][]{\relax
    {\noindent{\small$^{#1}$#2}}}

\author{Agnan Kessy$^{1}$, Alex Lewin$^{2}$,
and Korbinian Strimmer$^{3}$ 
\thanks{To whom correspondence should be addressed. Email: {\tt k.strimmer@imperial.ac.uk  } 
}
}

\maketitle

\myaddress[1]{Statistics Section, Dept. of Mathematics, Imperial College London, 
  South Kensington Campus, London SW7 2AZ, UK.\\}
\myaddress[2]{Dept. of Mathematics, Brunel University London, Kingstone Lane, Uxbridge UB8 3PH, UK.\\}
\myaddress[3]{Epidemiology and Biostatistics, School of Public Health,
  Imperial College London, Norfolk Place, London W2 1PG, UK.}

\newpage

\begin{abstract} 

Whitening, or sphering, is a common preprocessing step in statistical analysis 
to transform random variables to orthogonality.  However, due to rotational 
freedom there are infinitely many possible whitening procedures.  Consequently, 
there is a diverse range of sphering methods in use, for example based on principal 
component analysis (PCA), Cholesky matrix decomposition and 
{ zero-phase component analysis (ZCA)}, among others.

Here we provide an overview of the underlying theory and discuss 
five natural whitening procedures. Subsequently, we  
demonstrate that investigating the cross-covariance and the cross-correlation 
matrix between sphered and original variables allows to break the rotational 
invariance  and to identify optimal { whitening} transformations.  As a result 
we { recommend} two particular approaches: { ZCA-cor } whitening to 
produce sphered variables that are maximally similar to the original variables,
and { PCA-cor whitening to obtain sphered variables that
maximally compress the original variables.}

\vspace{1cm}


\end{abstract}

\newpage

\section{Introduction}

{

\emph{Whitening}, or \emph{sphering}, is a linear transformation 
that converts a $d$-dimensional random vector  $\bx=(x_1, \ldots, x_d)^T$ with mean 
$\expect(\bx)=\bmu = (\mu_1, \ldots, \mu_d)^T$ and positive definite $d \times d$ covariance matrix $\var(\bx)=\bSigma$ 
into a new random vector 
\begin{equation}
\bz = (z_1, \ldots, z_d)^T = \bW \bx 
\label{eq:whitening}
\end{equation}
 of the same dimension $d$ and with unit diagonal ``white'' covariance $\var(\bz)=\bI$.  The square $d \times d$ matrix $\bW$ is called the whitening matrix.
As orthogonality among random variables greatly simplifies multivariate data analysis
both from a computational and a statistical standpoint, whitening is a critically important tool, most often
employed in preprocessing but also as part of modeling \citep[e.g.][]{ZS09,HDF2015}.
 
Whitening can be viewed as a  generalization of \emph{standardizing} a random variable which is carried out by
\begin{equation} 
\bz = \bV^{-1/2} \bx \, ,
\end{equation} 
where the 
matrix $\bV = \diag(\sigma^2_1, \ldots, \sigma^2_d)$ 
contains the variances $\var(x_i)=\sigma^2_i$. This 
results in $\var(z_i)=1$
but it does not remove correlations.  Often, standardization and whitening transformations are
also accompanied by mean-centering of $\bx$ or $\bz$ to ensure $E(\bz) = 0$,
but this is not actually necessary for producing unit variances or a white covariance.

The whitening transformation defined in \eqcite{eq:whitening} 
requires the choice of a suitable whitening matrix~$\bW$.  Since  $\var(\bz) = \bI$
it follows that  $\bW \bSigma \bW^T = \bI$  and thus $\bW \, ( \bSigma \, \bW^T \bW ) = \bW$,
which is fulfilled if $\bW$ satisfies the condition
\begin{equation}
\bW^T  \bW = \bSigma^{-1} \, .
\label{eq:wconstraint}
\end{equation}
However, unfortunately, this constraint does not uniquely determine the whitening matrix $\bW$.  Quite the contrary, given $\bSigma$ there are in fact infinitely many possible matrices $\bW$ that all satisfy \eqcite{eq:wconstraint}, and each $\bW$
leads to a whitening transformation that produces
orthogonal but different sphered random variables.

This raises two important issues: first, 
 how to best understand the differences among the various sphering transformations, and second,
how to select an optimal whitening procedure for a particular situation.  Here, we propose to address
these questions by investigating the cross-covariance and cross-correlation matrix between $\bz$ and~$\bx$.
As a result, we identify five natural whitening procedures, of which we  recommend two particular approaches for general use.
}

\section{Notation and useful identities}

In the following, we will make use of a number of covariance matrix identities: 
 the decomposition  $\bSigma = \bV^{1/2} \bRho \bV^{1/2}$ of the covariance matrix into the correlation matrix $\bRho$ and the diagonal variance matrix $\bV$, and the eigendecomposition of the covariance matrix 
$\bSigma = \bU \bLambda \bU^T$ and the eigendecomposition of the correlation matrix
$\bRho =\bG \bTheta \bG^T$, where $\bU$, $\bG$ contain the eigenvectors and $\bLambda$, $\bTheta$ the eigenvalues
of $\bSigma$, $\bRho$ respectively.
We will frequently use $\bSigma^{-1/2} = \bU \bLambda^{-1/2} \bU^T $, the unique inverse matrix square root of $\bSigma$, as well as $\bRho^{-1/2} = \bG \bTheta^{-1/2} \bG^T$, the unique inverse matrix square root of the correlation matrix. 

 Following the standard convention we assume that the eigenvalues are sorted in order from largest to smallest value.
In addition, we  recall that by construction all eigenvectors  are defined only up to a sign, i.e.\ the
columns of $\bU$ and $\bG$ can be multiplied with a factor of $-1$ and the resulting matrix is still valid. Indeed, using different numerical algorithms and software will
often result in eigendecompositions with $\bU$ and $\bG$ showing diverse column signs.

\section{Rotational freedom in whitening}

{
The constraint \eqcite{eq:wconstraint} on the whitening matrix  does not fully identify  $\bW$
but allows for rotational freedom. This becomes apparent } by writing $\bW$ in its polar decomposition
\begin{equation}
\bW  = \bQ_1 \bSigma^{-1/2} \, ,
\label{eq:polardecomp}
\end{equation}
where $\bQ_1$ is an orthogonal matrix { with $\bQ_1^T \bQ_1 = \bI_d$.
 Clearly, $\bW$ satisfies \eqcite{eq:wconstraint} regardless of the choice of $\bQ_1$.}

This implies a geometrical interpretation of whitening as a combination of 
multivariate rescaling by $\bSigma^{-1/2}$ and rotation by $\bQ_1$.
It also shows that all whitening matrices $\bW$ have the same singular values 
$\bLambda^{-1/2}$, which follows from the singular value decomposition 
$\bW = (\bQ_1 \bU)  \bLambda^{-1/2} \bU^T$ with $\bQ_1 \bU$  orthogonal.
This highlights that the fundamental rescaling is via the square root of the eigenvalues $\bLambda^{-1/2}$.
Geometrically, the whitening transformation with $\bW = \bQ_1 \bU \bLambda^{-1/2} \bU^T$ 
is a rotation $\bU^T$ followed by scaling, 
possibly followed by another rotation (depending on the choice of $\bQ_1$).

Since in many situations it is desirable to work with standardized variables $\bV^{-1/2} \bx$ 
another useful decomposition of $\bW$ that also directly demonstrates the inherent rotational freedom is
\begin{equation}
\bW = \bQ_2  \bRho^{-1/2}  \bV^{-1/2} \, ,
\label{eq:Q2sphering}
\end{equation}
where $\bQ_2$ is a further orthogonal matrix { with $\bQ_2^T \bQ_2 = \bI_d$.
Evidently, this $\bW$ also satisfies the constraint of \eqcite{eq:wconstraint} regardless of the choice of $\bQ_2$.
}

In this view, with $\bW = \bQ_2  \bG \bTheta^{-1/2} \bG^T   \bV^{-1/2}$, the variables are first scaled by the square root of the diagonal variance
matrix, then rotated by $\bG^T$, then scaled again by the square root of the eigenvalues of the correlation matrix,
and possibly rotated once more (depending on the choice of $\bQ_2$).

{ For the above two representations 
to result in the same whitening matrix $\bW$ two different rotations
$\bQ_1$ and $\bQ_2$ are required.  These are linked by} $\bQ_1 = \bQ_2 \bA$ where the matrix $\bA = \bRho^{-1/2}  \bV^{-1/2} \bSigma^{1/2} = \bRho^{1/2}  \bV^{1/2} \bSigma^{-1/2} $ is itself orthogonal.
{ Since the eigendecompositions of the covariance and the correlation matrix are not readily related to each other,
the matrix $\bA$ can unfortunately not be further simplified.}

{
\section{Cross-covariance and cross-correlation}

For studying the properties of the different whitening procedures we will now focus on}
two particularly useful quantities, namely the  \emph{cross-covariance} and \emph{cross-correlation}
matrix 
between the whitened  vector $\bz$ and the original random vector $\bx$.
{ As it turns out, these are closely linked to the rotation matrices $\bQ_1$ and $\bQ_2$ encountered in
the above two decompositions of $\bW$ (\eqcite{eq:polardecomp} and \eqcite{eq:Q2sphering}).}

The cross-covariance matrix $\bPhi$ between $\bz$ and $\bx$ is given by 
\begin{equation}
\begin{split}
\bPhi = (\phi_{ij}) &= 
  \cov( \bz, \bx ) = \cov( \bW \bx, \bx ) \\     
   &= \bW \bSigma = \bQ_1  \bSigma^{1/2}  \, .
\end{split}
\label{eq:crosscov}
\end{equation}
{ Likewise, } the cross-correlation matrix is
\begin{equation}
\begin{split}
 \bPsi = (\psi_{ij}) & = 
 \cor( \bz, \bx ) = \bPhi \, \bV^{-1/2} \\
      & = \bQ_2 \bA \, \bSigma^{1/2}  \bV^{-1/2} = \bQ_2 \bRho^{1/2} \, .
\end{split}
\label{eq:crosscor}
\end{equation}
{ Thus, we find that the rotational freedom inherent in $\bW$,
which is represented by the matrices $\bQ_1$ and  $\bQ_2$,
is directly reflected in the corresponding cross-covariance $\bPhi$ and cross-correlation $\bPsi$ between $\bz$ and $\bx$. 
This provides the leverage that we will use to select and discriminate among whitening transformations
by appropriately  choosing or constraining 
$\bPhi$ or $\bPsi$. 

As can be seen from \eqcite{eq:crosscov} and  \eqcite{eq:crosscor}, both $\bPhi$ and $\bPsi$ are in general not symmetric, unless
$\bQ_1 = \bI$ or $\bQ_2 = \bI$, respectively.
Note that the diagonal elements of the
cross-correlation matrix $\bPsi$ need not be equal to 1.

Furthermore, since $\bx = \bW^{-1} \bz$
each $x_j$ is perfectly explained by a linear combination of the uncorrelated
$z_1, \ldots, z_d$, and hence the squared multiple correlation between $x_j$
and $\bz$ equals 1. Thus, the \emph{column} sum over the squared
cross-correlations $\sum_{i=1}^d \psi_{ij}^2$ is always~1. In matrix notation,
$\diag( \bPsi^T \bPsi ) = 
\diag( \bRho^{1/2}  \bQ_2^T  \bQ_2 \bRho^{1/2} ) =
\diag( \bRho ) = (1, \ldots, 1)^T$. 
In contrast, the \emph{row} sum of over the squared
cross-correlations  $\sum_{j=1}^d \psi_{ij}^2$ varies for different whitening procedures, and is, as we will see below, highly informative for choosing 
relevant transformations.
}

\section{Five natural whitening procedures}

In practical application of whitening there are { a handful of
sphering procedures that are} most commonly used
\citep[e.g.][]{LiZhang1998}.
{ Accordingly, in \tabcite{tab:sphering} we describe the properties of
 five whitening transformations, listing the respective
sphering matrix $\bW$, the associated rotation matrices $\bQ_1$ and $\bQ_2$,
and the resulting cross-covariances $\bPhi$  and cross-correlations $\bPsi$.
All five methods are natural whitening procedures arising from specific
constraints on $\bPhi$ or $\bPsi$, as we will show further below.}

\begin{table}[b]
\caption{Five natural whitening transformations and their properties.}
\centering
\begin{tabular}{l rrr l l}
\toprule
                & Sphering & Cross- & Cross- &  Rotation & Rotation \\

                & matrix & covariance & correlation &  matrix & matrix \\
                & $\bW$  & $\bPhi$    & $\bPsi$      & $\bQ_1$ & $\bQ_2$ \\
\midrule      
ZCA & $\bSigma^{-1/2}$                  & $\bSigma^{1/2}$                 & $\bSigma^{1/2} \bV^{-1/2}$         & $\bI$                 & $\bA^T$\\
PCA             & $\bLambda^{-1/2} \bU^T$           & $\bLambda^{1/2} \bU^T$          & $\bLambda^{1/2} \bU^T \bV^{-1/2}$  & $\bU^T$               & $\bU^T \bA^T$\\
Cholesky        & $\bL^T$                           & $\bL^T \bSigma$                 & $\bL^T \bSigma \bV^{-1/2}$         & $\bL^T \bSigma^{1/2}$ & $\bL^T \bV^{1/2} \bRho^{1/2}$ \\
ZCA-cor         & $\bRho^{-1/2} \bV^{-1/2}$         & $\bRho^{1/2} \bV^{1/2}$         & $\bRho^{1/2}$                      & $\bA$                 & $\bI$\\
PCA-cor         & $\bTheta^{-1/2} \bG^T \bV^{-1/2}$ & $\bTheta^{1/2} \bG^T \bV^{1/2}$ & $\bTheta^{1/2} \bG^T$              & $\bG^T \bA$           & $\bG^T$\\
\bottomrule
\end{tabular}
\label{tab:sphering}\\

\end{table}

The \emph{ZCA whitening} transformation employs the sphering matrix
\begin{equation}
\bW^{\text{ZCA}} = \bSigma^{-1/2}.
\end{equation}
where ZCA stands for ``zero-phase components analysis'' \citep{BS1997}.
This procedure is also known as \emph{Mahalanobis whitening}.
{With $\bQ_1 = \bI$} it is the unique sphering method 
with a symmetric whitening matrix.

{\emph{PCA whitening} is} based on scaled principal component analysis (PCA) and 
uses 
\begin{equation}
\bW^{\text{PCA}} = \bLambda^{-1/2} \bU^T
\label{eq:pcawhitening} 
\end{equation}
 \citep[e.g.][]{Friedman1987}.  
 This transformation first rotates the variables using the eigenmatrix of the covariance $\bSigma$
as is done in standard PCA. This results in orthogonal components, but with in general different variances.
 To achieve whitened data the rotated variables are then scaled by the square root of the eigenvalues $\bLambda^{-1/2}$. { PCA whitening } is probably the most widely 
applied whitening procedure due to its connection with PCA.

It can be seen that the PCA and ZCA whitening transformations are related by a rotation $\bU$,
so ZCA whitening can be { interpreted} as rotation followed by scaling followed by the rotation $\bU$ back to the original coordinate system.
The ZCA and the PCA sphering methods both naturally follow the polar decomposition of \eqcite{eq:polardecomp},
with $\bQ_1$ equal to $\bI$ and $\bU^T$ respectively.

Due to the sign ambiguity of eigenvectors $\bU$
the PCA whitening matrix given by \eqcite{eq:pcawhitening} is still not unique.
However, adjusting column signs in $\bU$ such that  $\diag(\bU) > 0$, i.e.\ that all diagonal elements are positive, results in the unique PCA whitening transformation with positive diagonal
cross-covariance~$\bPhi$ and cross-correlation~$\bPsi$ (cf. \tabcite{tab:sphering}).

{ Another widely known procedure is \emph{Cholesky whitening} which is based on
Cholesky factorization of the precision matrix $\bL \bL^T = \bSigma^{-1}$.  This leads to } the
 sphering matrix
\begin{equation} 
\bW^{\text{Chol}} = \bL^T
\end{equation} 
 where $\bL$
is the unique lower triangular matrix with positive diagonal values.  
 The same matrix $\bL$ can also be obtained 
from a QR decomposition of $\bW^{\text{ZCA}} = (\bSigma^{1/2} \bL) \,  \bL^T$.

{ A further approach is the \emph{ZCA-cor whitening} transformation, 
which is used, e.g., in the CAT (correlation-adjusted $t$-score) and CAR (correlation-adjusted marginal correlation)
variable importance and variable selection statistics
\citep{ZS09,AS2010,ZS2011,Zuber+2012}.  ZCA-cor whitening } employs 
\begin{equation}
\bW^{\text{ZCA-cor}} = \bRho^{-1/2} \bV^{-1/2} 
\label{eq:catcar-w}
\end{equation}
 as its sphering matrix.
{ It }arises from first 
standardizing the random variable by multiplication with $\bV^{-1/2}$ and subsequently 
employing ZCA whitening based on the correlation rather than covariance matrix.
{ The resulting whitening matrix $\bW^{\text{ZCA-cor}}$  differs 
from  $\bW^{\text{ZCA}}$, and unlike the latter it is in general} asymmetric.

In a similar fashion, { \emph{PCA-cor whitening} is conducted by
 applying PCA whitening to standardized variables. This approach uses }
\begin{equation}
\bW^{\text{PCA-cor}} = \bTheta^{-1/2} \bG^T \bV^{-1/2}
\label{eq:pcacor-w}
\end{equation}  
{ as its sphering matrix.}
Here, the standardized variables are rotated by the eigenmatrix of the correlation matrix,
followed by scaling using the correlation eigenvalues.
Note that $\bW^{\text{PCA-cor}}$ { differs } from $\bW^{\text{PCA}}$.

PCA-cor whitening  has the same relation to the ZCA-cor transformation as does 
PCA whitening to the ZCA transformation.  Specifically, 
ZCA-cor whitening can be interpreted as PCA-cor whitening followed by a rotation $\bG$ back to the frame of the standardized variables.
{ Both the ZCA-cor and the PCA-cor transformation naturally follow the  decomposition of \eqcite{eq:Q2sphering},
with $\bQ_2$ equal to $\bI$ and $\bG^T$ respectively.}

Similarly as in PCA whitening, the PCA-cor whitening matrix given by \eqcite{eq:pcacor-w} is subject to sign ambiguity of the eigenvectors in $\bG$.
As above, setting $\diag(\bG) > 0$ leads to the unique PCA-cor whitening transformation
with positive diagonal
cross-covariance~$\bPhi$ and cross-correlation~$\bPsi$ (cf. \tabcite{tab:sphering}).

Finally, we may also apply the Cholesky whitening transformation to standardized variables.
{ However, this does not lead to a new whitening procedure, as} the
 resulting sphering matrix remains 
identical to $\bW^{\text{Chol}}$ since the Cholesky factor of the inverse 
correlation matrix $\bRho^{-1}$ is $\bV^{1/2} \bL$,
{ and therefore }
$\bW^{\text{Chol-cor}} = (\bV^{1/2} \bL )^T \bV^{-1/2} = \bL^T = \bW^{\text{Chol}}$.

\section{Optimal whitening} 

We now demonstrate how an optimal sphering matrix $\bW$, and hence an optimal
whitening approach, can be identified by evaluating  suitable objective functions 
{ computed } from the cross-covariance $\bPhi$ and cross-correlation $\bPsi$.
{ Intriguingly, for each of the five natural whitening transforms 
listed in \tabcite{tab:sphering}
we find a corresponding optimality criterion.
 }

\subsection{ZCA-Mahalanobis whitening}

In many applications of whitening { it is desirable} 
to remove correlations { with minimal additional adjustment, with the aim } 
that the transformed { variable $\bz$} remains as similar as possible 
to the original vector $\bx$.

One possible implementation of this idea is to find the whitening transformation that 
minimizes the total squared distance between the original and whitened variables \citep[e.g.][]{EldarOppenheim2003}.  { Using mean-centered random vectors $\bz_c$ and $\bx_c$
with $\expect(\bz_c) = 0$ and $\expect(\bx_c)=0$} this least squares objective can be expressed as
\begin{equation}
\begin{split}
\expect\left( (\bz_c -\bx_c)^T  (\bz_c -\bx_c) \right) &= 
\trace(\bI) - 2 \, \expect\left( \trace\left( \bz_c \bx_c^T \right) \right) + \trace(\bSigma)\\ &= d - 2\trace(\bPhi) + \trace(\bV) \, .
\label{eq:lsdist}
\end{split}
\end{equation}
{ Since the dimension $d$ and sum of the variances $\trace(\bV) = \sum_{i=1}^d \sigma^2_i$ do not depend on the whitening matrix $\bW$ 
minimizing \eqcite{eq:lsdist} is equivalent to  maximizing} the trace of the cross-covariance matrix
\begin{equation}
\trace( \bPhi) = \sum_{i=1}^d \cov(z_i, x_i) =  \trace\left( \bQ_1 \bSigma^{1/2}  \right)   \equiv g_1(\bQ_1) \, .
\end{equation}

\begin{prop}
Maximization of $g_1(\bQ_1)$ uniquely determines the optimal whitening matrix to be
the symmetric sphering matrix $\bW^{\text{ZCA}}$.
\end{prop}

\begin{proof}
$g_1(\bQ_1) = \trace( \bQ_1 \bU \bLambda^{1/2} \bU^T ) = \trace( \bLambda^{1/2} \bU^T \bQ_1 \bU  )
 \equiv \trace( \bLambda^{1/2} \bB ) = \sum_i \bLambda_{ii}^{1/2} B_{ii}$ since $\bLambda$ is diagonal.
As $\bQ_1$ and $\bU$ are both orthogonal $\bB \equiv \bU^T \bQ_1 \bU$ is 
also orthogonal. This implies diagonal entries $B_{ii}\leq 1$, with equality 
signs for all $i$ occurring only if $\bB = \bI$, hence the maximum of $g_1(\bQ_1)$
is assumed at $\bB=\bI$, or equivalently at $\bQ_1 = \bI$.  From \eqcite{eq:polardecomp} it follows that
the corresponding optimal sphering matrix is
$\bW=\bSigma^{-1/2}=\bW^{\text{ZCA}}$.
\end{proof}
For related proofs see also \citet{Johnson1966},  \citet[][p. 412]{Gen1993} and
\citet[][p. 789]{Garthwaite2012}.

As a result, we { find} that \emph{ZCA-Mahalanobis whitening} is 
the unique procedure that \emph{maximizes the average 
cross-covariance} between each component of the whitened and original vectors.
Furthermore, with $\bQ_1=\bI$ it is also the unique
whitening procedure with a symmetric cross-covariance matrix $\bPhi$.

\subsection{ZCA-cor whitening}
{In the optimization using \eqcite{eq:lsdist} the underlying similarity measure 
is the cross-covariance between the whitened and original random variables.  
This results in an optimality criterion that depends on the variances and hence on the 
scale of the original variables.
An alternative scale-invariant objective 
can be constructed by comparing the centered} whitened variable with the centered \emph{standardized}  vector
$\bV^{-1/2} \bx_c$.  { This leads to the minimization of }
\begin{equation}
\expect\left( (\bz_c - \bV^{-1/2}\bx_c)^T  (\bz_c - \bV^{-1/2} \bx_c) \right) = 
2 d - 2\trace(\bPsi)  \,.
\end{equation}
{ Equivalently, we can maximize instead} the trace of the cross-correlation matrix
\begin{equation}
\trace(\bPsi) = 
 \sum_i^d \cor(z_i, x_i) = \trace\left(\bQ_2  \bRho^{1/2} \right) \equiv g_2(\bQ_2) \, .
\label{eq:catcar-func}
\end{equation}

\begin{prop}
Maximization of $g_2(\bQ_2)$  uniquely determines the whitening matrix  to be
the asymmetric sphering matrix $\bW^{\text{ZCA-cor}}$.   
\end{prop}

\begin{proof}
Completely analogous to Proposition 1, we can write 
$g_2(\bQ_2) = \trace( \bQ_2 \bG \bTheta^{1/2} \bG^T ) = \sum_i \bTheta_{ii}^{1/2} C_{ii}$
where $\bC \equiv \bG^T \bQ_2 \bG$ is orthogonal.
By the same argument as before it follows
that $\bQ_2 = \bI$ maximizes $g_2(\bQ_2)$.
From \eqcite{eq:Q2sphering} it follows that $\bW  = \bRho^{-1/2} \bV^{-1/2} = \bW^{\text{ZCA-cor}}$.
\end{proof}

As a result, we identify \emph{ZCA-cor whitening} as the unique procedure 
{ that  \emph{ensures}} \emph{that the components of the whitened vector
$\bz$ remain maximally correlated with the corresponding
components of the original variables $\bx$}.
In addition, with $\bQ_2=\bI$ it is also the unique whitening transformation
exhibiting a symmetric cross-correlation matrix~$\bPsi$.

\subsection{PCA whitening}

Another frequent { aim} in whitening is the generation of new 
{ uncorrelated} variables $\bz$
that are useful for dimension reduction and { data compression.  In other words, we would like to 
construct components $z_1, \ldots, z_d$ such that the first few components in $\bz$ represent as much as possible
the variation present in the all original variables $x_1, \ldots, x_d$.

One way to formalize this is to use 
the row sum of squared cross-covariances 
$\phi_i = \sum_{j=1}^d \phi_{ij}^2 = \sum_{j=1}^d \cov(z_i, x_j)^2$ 
between each individual $z_i$
and all $x_j$ as a measure of how effectively each $z_i$ integrates, or compresses, the original variables.
Note that here, unlike in ZCA-Mahalanobis whitening,
the objective function links
each component in $\bz$ simultaneously with all components in $\bx$.
In vector notation the $\phi_i$ can be more elegantly written as}
\begin{equation}
(\phi_1, \ldots, \phi_d)^T = \diag\left( \bPhi \bPhi^T \right) = \diag\left(\bQ_1 \bSigma \bQ_1^T \right) 
\equiv \bh_1(\bQ_1)\,.
\end{equation}
{ Our aim is to find a whitened vector $\bz$ such that the $\phi_i$ are maximized 
with $\phi_i \geq \phi_{i+1}$. }

\begin{prop}
Maximization of $\bh_1(\bQ_1)$ subject to monotonically decreasing $\phi_i$ is achieved by the whitening matrix 
$\bW^{\text{PCA}}$.
\end{prop}
{
\begin{proof}
The vector
$\bh_1(\bQ_1)$ can be written as $\diag( \bQ_1 \bSigma \bQ_1^T) = \diag(\bQ_1 \bU\bLambda \bU^T  \bQ_1^T )$.
Setting $\bQ_1=\bU^T$ we arrive at $\bh_1(\bU^T) = \diag(\bLambda) $, i.e. for this choice
the $\phi_i$ corresponding to each component of $z_i$ are equal to the corresponding
eigenvalues of $\bSigma$.  As the eigenvalues are already sorted in decreasing order,
we  find (cf. \tabcite{tab:sphering}) that whitening with $\bW^{\text{PCA}}$
leads to a sphered variable $\bz$ with monotonically decreasing $\phi_i$.
For general $\bQ_1$ the $i$-th element of  $\bh_1(\bQ_1)$ is $\sum_j \bLambda_{jj} D_{ij}^2$ 
where $\bD \equiv \bQ_1 \bU$ is orthogonal.  
This is maximized when $\bD = \bI$, or equivalently, $\bQ_1=\bU^T$.
\end{proof}

As a result, \emph{PCA whitening} is singled out as the unique sphering procedure
that \emph{maximizes the integration, or compression, of all components of the
 original vector $\bx$ in each component 
of the sphered vector $\bz$
 based on the cross-covariance~$\bPhi$ as underlying measure.}
Thus,} the fundamental property of PCA that principal components
are optimally ordered with respect to dimension reduction \citep{Jolliffe2002}
{ carries over also to PCA whitening.}

\subsection{PCA-cor whitening}

For reasons { of scale-invariance } we prefer 
to optimize cross-correlations rather than cross-covariances for whitening with 
compression in mind.  This leads to  { the row sum of squared cross-correlation 
$\psi_i = \sum_{j=1}^d \psi_{ij}^2 = \sum_{j=1}^d \cor(z_i, x_j)^2$
as measure of integration and compression, and correspondingly to the objective function}
\begin{equation}
(\psi_1, \ldots, \psi_d)^T =  \diag\left( \bPsi \bPsi^T \right) = \diag\left(\bQ_2 \bRho \bQ_2^T \right) 
= \bh_2(\bQ_2)\, .
\label{eq:pcacor-func}
\end{equation}

\begin{prop}
{ Maximization of $\bh_2(\bQ_2)$ subject to monotonically decreasing $\psi_i$  is achieved by using}
$\bW^{\text{PCA-cor}}$ as the sphering matrix.
\end{prop}

\begin{proof}
Analogous to the proof of Proposition~3 we find $\bQ_2=\bG^T$ to { yield optimal and decreasing
$\psi_i$ } and 
with \eqcite{eq:Q2sphering} we arrive at
$\bW=\bTheta^{-1/2} \bG^T  \bV^{-1/2} = \bW^{\text{PCA-cor}}$.
\end{proof}

{
Hence, the \emph{PCA-cor whitening} transformation is the unique transformation that
\emph{maximizes the integration, or compression, of all components of the
 original vector $\bx$ in each component 
of the sphered vector $\bz$
 employing the cross-correlation~$\bPsi$ as underlying measure}.

}

\subsection{Cholesky whitening}

{ Finally, we investigate the connection between Cholesky whitening and corresponding characteristics of the 
cross-covariance and cross-correlation matrices.
Unlike the other four whitening methods listed in \tabcite{tab:sphering}, 
which result from optimization, Cholesky whitening is due to a symmetry constraint.

Specifically, the whitening matrix $\bW^{\text{Chol}}$  leads to a cross-covariance matrix $\bPhi$ that is lower-triangular with positive diagonal elements as well as to a cross-correlation matrix $\bPsi$ with the same properties.
This is a consequence of the Cholesky factorization  with $\bL$ being subject to the same 
constraint.   Crucially, as $\bL$ is unique the converse argument is valid as well, and hence \emph{Cholesky whitening} is 
the unique whitening procedure that  \emph{results from lower-triangular positive diagonal cross-covariance and cross-correlation matrices.}

A consequence of using Cholesky factorization for whitening  
is that we implicitly assume an ordering of the variables. 
This can be useful specifically in time course analysis
to account for auto-correlation \citep[cf.][and references therein]{Pourahmadi2011}.
}

\section{Application} 

{
\subsection{Data-based whitening}

In the sections above, we have discussed the theoretical background of whitening
in terms of random variables $\bx$ and $\bz$ and using the population covariance $\bSigma$
to guide the construction of a suitable sphering matrix $\bW$.

In practice, however, we frequently  need to whiten data rather than random variables.
In this case we have an $n \times d$ data matrix  $\bX = (x_{ki}) $ whose \emph{rows} are assumed to 
be drawn from a distribution with expectation $\bmu$ and covariance matrix $\bSigma$.
In this setting the transformation of \eqcite{eq:whitening} from original to whitened data matrix
becomes  $\bZ = \bX \bW^T$.

A further complication is that the covariance matrix $\bSigma$ is often unknown.
Accordingly, it needs to be learned from data, either from $\bX$ or from another suitable data set,
yielding a covariance matrix estimate $\widehat{\bSigma}$.  Typically, for large sample size
$n$ and small dimension~$d$  the standard unbiased empirical covariance $\bS = (s_{ij})$ with $s_{ij} = \frac{1}{n-1} \sum_{k=1}^n (x_{ki}-\bar{x}_i) (x_{kj}-\bar{x}_j) $ is used.
In high-dimensional cases with $p > n$ the empirical estimator breaks down, and the covariance matrix needs to be estimated by a suitable regularized method instead \citep[e.g.][]{SS05c,Pourahmadi2011}.   
Finally, from the spectral decomposition of the estimated covariance or corresponding correlation matrix we then obtain the desired estimated whitening matrix~$\widehat{\bW}$.

\subsection{Iris flower data example}
}

\begin{table}[t]
\caption{Whitening transforms applied to the iris flower data set.}
\centering
\begin{tabular}{l rrrrr }
\toprule
                      & ZCA & PCA & Cholesky & ZCA-cor & PCA-cor \\
\midrule
$\widehat{\cor}(z_1, x_1)$ & 0.7137 & 0.8974  & 0.3760 & 0.8082 & 0.8902 \\
$\widehat{\cor}(z_2, x_2)$ & 0.9018 & 0.8252  & 0.8871 & 0.9640 & 0.8827 \\
$\widehat{\cor}(z_3, x_3)$ & 0.8843 & 0.0121  & 0.2700 & 0.6763 & 0.0544 \\           
$\widehat{\cor}(z_4, x_4)$ & 0.5743 & 0.1526  & 1.0000 & 0.7429 & 0.0754 \\
\midrule
$\trace (\widehat{\bPhi})$ & {\bf 2.9829} & 1.2405 & 1.9368 & {\it 2.8495} & 1.2754 \\
$\trace (\widehat{\bPsi})$ & {\it 3.0742} & 1.8874 & 2.5331 & {\bf 3.1914} & 1.9027 \\
$\max \diag(\widehat{\bPhi}\widehat{\bPhi}^T)$ & 3.1163 & {\bf 4.2282} & 3.9544 & 1.7437       & {\it 4.1885} \\
$\max \diag(\widehat{\bPsi}\widehat{\bPsi}^T)$ & 1.9817       & {\it 2.8943}       & 2.7302 & 1.0000       & {\bf 2.9185}\\
\bottomrule
\end{tabular}
\\Bold font indicates best whitening transformation, and italic font the second best method for each considered criterion (lines 5--8).
\label{tab:iris}\\
\end{table}

For an illustrative comparison of the five natural whitening transforms
{ discussed in this paper and listed in \tabcite{tab:sphering}
we applied them on the well-known iris flower data set of Anderson reported in \citet{Fisher1936}},
which comprises
$d=4$ correlated variables ($x_1$: sepal length,  
$x_2$: sepal width, 
$x_3$: petal length,
$x_4$: petal width)
and  $n=150$ observations.

The results are shown in \tabcite{tab:iris} with all estimates 
based on the empirical covariance $\bS$. For the PCA and PCA-cor whitening transformation we have
set $\diag(\bU)>0$ and $\diag(\bG)>0$, respectively.
The upper half of  \tabcite{tab:iris} shows the estimated cross-correlations between each component of the  whitened and original vector for the five methods, and the lower
half the values of the various objective functions discussed above.

As expected, the ZCA and the ZCA-cor whitening produce sphered variables that are most correlated  to the original data on a component-wise level, with the former achieving the best fit for the covariance-based  and the latter for the correlation-based objective.

In contrast, the PCA and PCA-cor methods are best at producing whitened variables
that are maximally simultaneously linked with all components  of the original variables.  Consequently, as can be seen from the top half of \tabcite{tab:iris}, 
for PCA and PCA-cor whitening only the first two components { $z_1$ and $z_2$ are highly correlated with their respective counterparts
$x_1$ and $x_2$, whereas the subsequent pairs $z_3, x_3$ and $z_4, x_4$ are effectively uncorrelated.
Furthermore, the last line of \tabcite{tab:iris} shows that PCA-cor whitening achieves
higher maximum total squared correlation } of the first component $z_1$ with all components of $\bx$
than PCA whitening, indicating better compression.

{ Interestingly, Cholesky whitening always assumes third place in the rankings, either behind ZCA and ZCA-cor whitening, or behind PCA and PCA-cor whitening.
Moreover, it is the only approach where by construction one pair ($z_4, x_4$)  perfectly correlates between whitened and original data.}

\section{Conclusion}

{ In this note we have investigated linear transformations for whitening of random variables.
These methods are commonly employed in data analysis for preprocessing and to
facilitate subsequent analysis.

In principle, there are infinitely many possible whitening procedures
all satisfying the fundamental constraint of \eqcite{eq:wconstraint} for the underlying
whitening matrix.}
 However, as we have demonstrated here, the rotational freedom inherent in whitening can be
broken by considering  cross-covariance $\bPhi$ and
cross-correlations $\bPsi$ between whitened and original variables.

Specifically, we have { studied} five natural whitening transforms,
cf.~\tabcite{tab:sphering}, all of which
can be interpreted as either optimizing a suitable function of  $\bPhi$ or $\bPsi$, or 
satisfying a symmetry constraint on $\bPhi$ or $\bPsi$.  
{ As a result, this not only leads to a better understanding of the differences among whitening
methods, but also enables an informed choice.

In particular, selecting a suitable whitening transformation depends on the context of application,
specifically whether minimal adjustment or compression of data is desired.  In the former  the whitened variables
remain highly correlated to the original variables, and thus maintain their original interpretation.
This is advantageous for example in the context of variable selection where one would like to understand
the resulting selected submodel.  In contrast, in a compression 
context the whitened variables by construction bear no interpretable relation to the original data but instead
reflect their intrinsic effective dimension.

In general, we advocate using scale-invariant optimality functions and thus recommend
using cross-correlation $\bPsi$ as a basis for optimization. 
Consequently, we  particularly endorse two specific whitening approaches.
If the aim is to obtain sphered variables that are maximally similar to the original 
ones, we suggest to employ  the ZCA-cor whitening procedure 
of \eqcite{eq:catcar-w}.  Conversely,  if maximal compression 
is desirable we recommend to use the PCA-cor whitening approach of \eqcite{eq:pcacor-w}.}


\bibliographystyle{apalike}
\bibliography{preamble,stats,strimmer}

\end{document}